\documentclass[11pt]{article}
\usepackage{amsmath, amssymb, amsthm, xparse, xcolor, algorithm, dsfont, algpseudocode, bbm, marginnote}
\usepackage{thm-restate}
\usepackage[parfill]{parskip}
\usepackage[most]{tcolorbox}
\usepackage[margin=1in]{geometry}
\usepackage{dsfont}
\usepackage[hidelinks]{hyperref}
\usepackage[normalem]{ulem}
\usepackage{cleveref}
\usepackage{mathtools}
\usepackage{tikz}
\usetikzlibrary{decorations.pathreplacing, calc}

\newtheorem{theorem}{Theorem}
\newtheorem{lemma}[theorem]{Lemma}

\theoremstyle{definition}
\newtheorem{definition}[theorem]{Definition}

\newtheorem{question}[theorem]{Question}

\NewDocumentCommand{\ot}{g}{
    \IfNoValueTF{#1}{\tilde{O}}{\tilde{O}(#1)}
}

\DeclareMathOperator{\poly}{poly}
\DeclareMathOperator{\polylog}{polylog}

\DeclareMathOperator{\supp}{Supp}

\newcommand{\calS}{\mathcal{S}}

\newcommand{\calA}{\mathcal{A}}

\newcommand{\calC}{\mathcal{C}}

\newcommand{\calI}{\mathcal{I}}

\newcommand{\calU}{\mathcal{U}}

\newcommand{\se}{\textsf{se-consrt}}

\newcommand{\ddec}{d_{\textsf{decide}}}

\definecolor{amber}{rgb}{1.0, 0.49, 0.0}

\newcommand{\br}{\boldsymbol{r}}

\newcommand{\bq}{\boldsymbol{q}}

\newcommand{\ind}[1]{\mathbbm{1}\{#1\}}

\title{Approximate Nearest Neighbor in Ultra-High Dimensional $\ell_\infty$}
\author{Nathan White \medskip \\ University of Pennsylvania \\ \href{mailto:nathanlw@cis.upenn.edu}{\texttt{nathanlw@cis.upenn.edu}} \and Tian Zhang \medskip \\ University of Pennsylvania \\ \href{mailto:tianzh@cis.upenn.edu}{\texttt{tianzh@cis.upenn.edu}}}
\date{\today}

\begin{document}
\maketitle

\begin{abstract}
    We study the approximate nearest neighbor problem under $\ell_\infty$ in the ultra-high dimensional setting where the dimension $d$ is significantly larger than the number of points $n$. Thus, we desire data structures with no dependence on $d$ in the query time. [Herold-Nanongkai-Spoerhase-Varma-Wu, SoCG 2025] introduce this problem and give data structures in $\ell_p$: for $p=1,2$, they give $(1+\varepsilon)$-approximation data structures with space $\tilde{O}(n\log d/\text{poly}(\varepsilon))$ and query time $\tilde{O}(n/\text{poly}(\varepsilon))$. Since any data structure must have query time $\Omega(\min \{n,d\})$, this query time is nearly tight. However, their results are inefficient for $\ell_\infty$, with query time $\Omega(nd)$.

In order to handle the challenges of $\ell_\infty$, we introduce a notion of \textit{subset embeddings}, which embed points by simply selecting a subset of dimensions. In particular, we show one may preserve all pairwise distances of an $n$ point dataset up to a factor of $O(c)$ by computing distances on only $n^{1+1/c}$ coordinates. We also show a matching lower bound: for any $c > 1$, there exists a set of $n$ points in $\mathbb{R}^{d}$ such that any subset embedding for the set with approximation $c$ must have at least $n^{1+\Omega(1/c)}$ coordinates.

Using our subset embeddings, we give data structures for approximate nearest neighbor in $\ell_\infty$ which have query time independent of $d$ and space that depends only logarithmically on $d$. Specifically, for any approximation factor $c \geq 1$, we give data structures with space $O(n^2\log d)$, query time $\tilde{O}(n^{1+1/c})$, and approximation $O(c\log\log n)$. This result implies a $O(\log n)$-approximation with query time $\tilde{O}(n)$, which matches the query time lower bound up to $\text{polylog} n$ factors. 

Finally, we give another data structure for the approximate nearest neighbor under $\ell_\infty$ with the same space and query time as our subset embedding approach, but with approximation $O(c^{\log_2 3}) \approx O(c^{1.58})$. This allows us to achieve $O(1)$-approximation with query time e.g.~$n^{1.01}$
\end{abstract}

\newpage
\section{Introduction}
The Approximate Nearest Neighbor (ANN) problem is defined as follows: the input is a dataset $X$ of $n$ points in a metric space $(\calU, d_{\calU})$, and a desired approximation $c$.
Our goal is to preprocess the dataset $X \subset \calU$ into a data structure.
Upon the arrival of a previously unseen query point $q \in \calU$, we must output a point $\hat{p}$ in $X$ whose distance to $q$ is within a $c$ factor of the distance between $q$ and its nearest neighbor $p^*$ in $X$, i.e.~$d_{\calU}(q,\hat{p}) \leq c \cdot d_{\calU}(q, p^*)$.
Besides the approximation factor $c$, the performance of ANN data structures is also measured by storage space and query time.

In this work, the universe $\calU$ is always $\mathbb{R}^d$. The dimension $d$ plays an important role in how one can approach the ANN problem.
In low dimensions where $d = o(\log n)$, space and query time that are exponential in $d$ are affordable.
Efficient data structures, with query time that depends logarithmically on $n$, are known for all metric spaces in this regime \cite{clarkson_kenneth_nns, krauthgamer_navigating_nets, karger_nn_growth_restricted, beygelzimer_cover_tree}.
In the high-dimensional regime where $\omega(\log n) \leq d \leq o(n)$, algorithms for the low-dimensional setting suffer from the ``curse of dimensionality", where $2^d$ is prohibitively large.
For the high-dimensional regime, researchers have designed efficient data structures for various metric spaces \cite{IM98, indyk2001approximate, andoni2008earth, andoni2017approximate}.
These data structures typically have space or query time that is polynomial in the dimension $d$, and their query time is also sublinear in $n$.

\paragraph{Ultra-High Dimensional ANN} We consider the ultra-high dimensional regime, first studied in \cite{mpii-ultra}, where the dimension $d$ far exceeds the number of points, i.e.~$d \gg n$. 
This setting is motivated by the recent explosion in the complexity of data, for example genetic data that lie in spaces with tens of millions of dimensions \cite{stanaway2019emerge}.
The goal of ultra-high dimensional ANN is to design data structures with very little dependence on $d$ in the storage space and query time, given that a cost polynomial or even linear in $d$ is prohibitively large.
In particular, we desire data structures with \textit{no} dependence on $d$ in the query time and only logarithmic dependence on $d$ in space.

These restrictions pose some unique challenges: we can only read a small fraction of the query point and cannot even store any full point in our data structure!\footnote{We still allow preprocessing time to be linear in $d$, as this is necessary to be able to read the whole dataset.}
For example, in ultra-high dimensional Euclidean space, applying the well-known Johnson-Lindenstrauss (JL) Lemma to reduce dimensions is infeasible: even if JL were able to reduce the dimension in preprocessing, it would need to read the entire query vector, resulting in query time linear in $d$ \cite{johnson1984extensions}.
Therefore, as with going from the low-dimensional regime to the high-dimensional regime, new techniques are required for ANN data structures in ultra-high dimensions.

For $\ell_1$ and $\ell_2$, \cite{mpii-ultra} gives $(1+\varepsilon)$-approximate data structures with space $\ot{n \log d / \poly(\varepsilon)}$ and query time $\ot{n / \poly(\varepsilon)}$, which is nearly optimal since any data structure must have query time $\Omega(\min\{n, d\})$\footnote{This lower bound is shown in \cite{mpii-ultra}; we also include a proof in Lemma \ref{lem:qt-lower-bound} for completeness.}.
For general $\ell_p$, they achieve a $(1+\varepsilon)$-approximation with space $\ot{n^2 \log d \cdot (\log \log n / \varepsilon)^p}$ and query time $\ot{n \cdot (\log \log n / \varepsilon)^p}$, for $\varepsilon=O(1/p)$.
Their techniques rely on importance sampling over the coordinates based on pairwise distances, which works well when $p$ is small.
However, for large $p$, the space and query time grow exponentially in $p$, making these data structures inefficient for large $p$.
In particular, their query time for $\ell_\infty$ is $\Omega(nd)$.

\paragraph{$\ell_\infty$ Metrics}
$\ell_\infty$ is an important and well-studied metric, in part because every metric space of $n$ points can be isometrically embedded into $\ell_\infty$ of dimension $n$.
Perhaps unsurprisingly, then, compressing vectors in $\ell_\infty$ is notoriously hard to achieve while maintaining distances (i.e.~via sketching): to achieve approximation $k$, any $\ell_\infty$ sketch must have size $\Omega(n/k^2)$.
This, at first glance, presents a challenge for ultra-high dimensional ANN in $\ell_\infty$: since no dataset point can be stored in its entirety, the data structure must store a ``compressed'' representation of the points.
Fortunately, our notion of subset embeddings (see Section \ref{sec:tech-overview} and Section \ref{sec:subset-embed}) will help us overcome this difficulty.

\subsection{Main Results}
One might desire a $(1+\varepsilon)$-approximation data structure, for any $\varepsilon > 0$.
Unfortunately, we show this is impossible with $o(d)$ query time: any data structure with approximation better than 3 requires $\Omega(d)$ query time.
This lower bound is tight; we note in Observation \ref{obs:3-approx} a 3-approximation with query time $O(n^2)$ (and no dependence on $d$).

\begin{theorem}
\label{thm:lb-randomized-alg}
    For any $c < 3$, any randomized data structure for $c$-approximate nearest neighbor over $\ell_{\infty}^d$ with success probability at least $0.5001$ for any dataset and query must have query time $\Omega(d)$.
\end{theorem}
We prove Theorem \ref{thm:lb-randomized-alg} in Section \ref{sec:lower-bound-proofs}.

Our main results are two data structures for ultra-high dimensional ANN in $\ell_\infty$, with different trade-offs between query time and approximation.
\begin{theorem}
    \label{thm:small-query-time}
    For any $c\geq 1$, there exists a data structure for approximate nearest neighbor over $\ell_\infty^d$ with
    \begin{itemize}
        \item Approximation $O(c \cdot \log\log n)$;
        \item Query time $\ot{n^{1+1/c}}$; and
        \item Storage space $O(n^{2}\log d)$
    \end{itemize}
\end{theorem}

\begin{theorem}
    \label{thm:c-258}
    For any $c \geq 1$, there exists a data structure for approximate nearest neighbor over $\ell_\infty^d$ with
    \begin{itemize}
        \item Approximation $3\cdot c^{\log 3} \approx O(c^{1.58})$;
        \item Query time $\ot{n^{1+1/c}}$; and
        \item Storage space $O(n^2\log d)$
    \end{itemize}
\end{theorem}

These results allow for a spectrum of performance: we can achieve an $O(1)$-approximation with close to linear query time $O(n^{1.01})$, or obtain a $O(\log n)$-approximation with query time $\ot{n}$.
Note that the query time of this latter data structure is tight, up to $\polylog n$ factors, for any data structure with finite approximation (see Lemma \ref{lem:qt-lower-bound}).

\subsection{Technical Overview}
\label{sec:tech-overview}
Our algorithms are based on a simple yet very useful observation: for every point $x\in X$, there is a set $S_x$ of only $n$ coordinates such that for all $y\in X$, there exists an $i\in S_x$ with $|x_i - y_i| = \|x - y\|_\infty$. Then, we define the \textit{deciding distance} between a point $q$ and a dataset point $x\in X$ as
\[
    \ddec(q,x) = \max_{i\in S_x} |q_i - x_i|.
\]
We can show that, for any $q\in \mathbb{R}^d$, the point $x\in X$ which minimizes $\ddec(q,x)$ is a 3-approximation to the true nearest neighbor (see the proof of Lemma \ref{lem:3-approx}).
This then allows us to prove the following result by storing $S_x$ for all $x \in X$ and then simply computing $\ddec(q,x)$ for each point $x$ at query time.

\begin{restatable}{observation}{threeapprox}
    \label{obs:3-approx}
    There exists a data structure for 3-approximate nearest neighbor over $\ell_\infty$ with
    \begin{itemize}
        \item Query time $O(n^2)$;
        \item Storage space $O(n^{2}\log d)$
    \end{itemize}
\end{restatable}

To improve query time beyond $O(n^2)$, we generalize the concept of deciding coordinates into what we call \textit{subset embeddings}.
A subset embedding is a set of coordinates $S \subseteq [d]$ that preserves pairwise distances for the entire dataset up to a distortion factor $c$.
Formally, $S$ is an $\ell_\infty$ subset embedding with distortion $c$ if for all $x, y \in X$
\[
\|x_S - y_S\|_\infty \ge \frac{1}{c} \|x - y\|_\infty
\]
where $x_S$ is the restriction of vector $x$ to the coordinates in $S$.
Note that, because $x_S$ is a restriction of $x$, $\|x_S - y_S\|_\infty \leq \|x-y\|_\infty$ for all $S\subseteq [d]$.

A natural question is how small we can make subset embeddings with a given distortion.

\begin{question}
    \label{question:subset}
    For a given $c \geq 1$ and an arbitrary dataset, what is the smallest number of coordinates needed to preserve all pairwise $\ell_\infty$ distances in the dataset up to a factor of $c$?
\end{question}
We observe a lower bound for this question: there exists a dataset for which any finite distortion subset embedding has size $\Omega(n)$ (Lemma \ref{lem:lb-finite-dist}).
For upper bounds, since there are only $O(n^2)$ many pairs of points, there is always an isometric subset embedding with size $O(n^2)$.
Thus, for any $c$, the optimal size must lie between $\Omega(n)$ and $O(n^2)$.

We improve both these upper and lower bounds, and answer Question \ref{question:subset} with matching upper and lower bounds: for any $c\geq 1$, we show there is a subset embedding with distortion $O(c)$ and size at most $n^{1+1/c}$, which is tight up to constants in the distortion factor.

\begin{restatable}{theorem}{subsetemb}
\label{thm:subset-embed-imp}
    For any $n$ point dataset $X \subseteq \mathbb{R}^{d}$ and distortion parameter $c\geq 1$, there exists an $\ell_\infty$ subset embedding for $X$ with distortion $O(c)$ consisting of $n^{1+1/c}$ coordinates.
\end{restatable}

Unlike the Johnson-Lindenstrauss transform or a tree embedding, a subset embedding may only select a subset of the original coordinates, and a coordinate selected for one pair of points may be useless for every other pair. Our selection procedure is greedy: as long as some pair of points is not yet approximated within the target factor, we take the furthest such pair and select a coordinate realizing its distance.
This guarantees the distortion by definition, so the entire difficulty lies in bounding the number of coordinates selected. To do so, we connect the procedure to a witness graph over the dataset, where each edge is the pair of points responsible for one selected coordinate. Interestingly, we are able to show that the greedy order forces this graph to have girth $c$. Since a graph over $n$ nodes with girth $c$ has at most $n^{1+O(1/c)}$ edges \cite{alon2002moore}, so does the witness graph, and therefore the size of our subset embedding is bounded by  $n^{1+O(1/c)}$.

We also show that this construction is tight up to constants in the distortion.
\begin{theorem}
    \label{thm:subset-embed-lower-bound}
    For $c,n$ larger than fixed constants, there exists a dataset of $n$ points for which any $\ell_\infty$ subset embedding of distortion at most $c$ must include $n^{1+\Omega(1/c)}$ coordinates.
\end{theorem}
The proof of Theorem \ref{thm:subset-embed-lower-bound} uses similar ideas to the upper bound: from any graph with $n$ nodes, $m$ edges and girth $\lambda=\Theta(c)$, we construct a dataset of $n$ points such that any subset embedding with distortion at most $c$ must include $m$ coordinates.
The lower bound then follows from the classic result that there are graphs with girth $\lambda$ and $n^{1+\Omega(1/\lambda)}$ edges \cite{erdos1963regulare}.
See Section \ref{sec:subset-embed-lower} for the full proof.

Given subset embeddings, we now have two methods for constructing ultra-high dimensional $\ell_\infty$ ANN data structures.
First, we could simply compute a subset embedding with distortion $O(c)$ which uses $d'=n^{1+1/c}$ coordinates, store these coordinates, and apply the high-dimensional ANN data structure of \cite{indyk2001approximate}.
This approach yields query time $O(d'\polylog n) = \ot{n^{1+1/c}}$ and approximation $O(c\log \log d') = O(c\log\log n)$, from the guarantees of the data structure of \cite{indyk2001approximate}.\footnote{This direct approach, however, has space which exceeds $O(n^2\log d)$; as such, a few additional steps are needed to reduce the space (see Section \ref{sec:subset-ann} for details).} This approach will allow us to prove Theorem \ref{thm:small-query-time}.

Second, we can use the $O(n^2)$ sized subset embedding, with distortion $1$, implied by deciding distance.
However, to reduce the query time from $O(n^2)$, we employ a divide-and-conquer technique to repeatedly filter far points, at the cost of additional approximation.
This allows us to prove Theorem \ref{thm:c-258}; see Section \ref{sec:c-258} for details.

\subsection{Related Work}
\paragraph{High-Dimensional $\ell_\infty$ ANN} \cite{indyk2001approximate} gives an ANN data structure for $\ell_\infty$ with approximation $O(\log_{1+\rho} \log d)$, space $d \cdot n^{1+\rho} \cdot \polylog n$ and query time $d \cdot \polylog n$, for a parameter $\rho > 0$.
Of particular note, this implies an $O(\log \log d)$-approximation with space $O(n^{1.01}d)$ and query time $d\cdot\polylog n$.
\cite{linflower2008} show that, for data structures which are decision trees or have constant cell-probe complexity, approximation $O(\log \log d)$ is the best possible with polynomial space and sublinear query time in $n$. 

\paragraph{Embeddings} Embeddings are popular techniques for solving geometric problems. Besides JL, readers who are familiar with other embeddings may wonder if they are helpful for the ultra-high dimensional ANN problem.
In fact, a na\"ive attempt with any embedding that relies on the entire query requires $\Omega(d)$ time to compute, which is prohibitively large in the ultra-high dimensional setting.
For example, tree embeddings in \cite{bartal1996probabilistic,frt2003tight} embed any metric space into a tree metric space of depth upper-bounded by logarithmic aspect ratio $\log \Delta$ (thus into $\ell_1$ where each embedding vector has only $\log \Delta$ effective entries, which is independent of $d$).
However, computing the tree embeddings requires computing distances in the original space, which costs $\Omega(d)$.
Moreover, the tree embeddings are data-dependent, meaning their distance distortion guarantees are between pairs only in the preprocessed dataset.
We will introduce embeddings (subset embeddings, see Section \ref{sec:subset-embed}) which overcome these two challenges; the time to compute the embedding for a query does not depend on dimensions and it guarantees an approximation for ANN. 

\paragraph{Feature Selection} 
There is also the related technique of feature selection.
The goal of feature selection is to select a small subset of the most relevant features (coordinates/dimensions) from a high-dimensional dataset while preserving the essential structure of the data for a specific downstream problem.
For $k$-means clustering, efficient and accurate deterministic \cite{Boutsidis_2013, cohen2015dimensionalityreductionkmeansclustering} and randomized \cite{randomfeatureselection1,boutsidis2014randomizeddimensionalityreductionkmeans,cohen2015dimensionalityreductionkmeansclustering} algorithms have been developed\footnote{These feature selection techniques for clustering preserve only the total clustering cost, and so do not directly yield good nearest neighbor guarantees.
However, \cite{mpii-ultra} (Appendix F) shows how to use those techniques designed for clustering in ultra-high dimensional ANN in $\ell_2$.}.
For ultra-high dimensional ANN, the importance sampling over coordinates in \cite{mpii-ultra} can be viewed as a randomized feature selection for $\ell_p$ with small $p$.
Our subset embeddings, on the other hand, deterministically select a subset of the features for $\ell_\infty$; see Section \ref{sec:subset-embed} for the full details of our embeddings.

\section{Preliminaries}

Throughout, we work in the word-RAM model, and thus assume for all  points $x$ given to the algorithm, $x_i$ can be written in $O(1)$ words for all $i$.

We assume that the data structure stores a list $\{i_1,\ldots,i_k\}\subseteq [d]$ of coordinates, and on query $q$, the query algorithm is passed only the values of $q$ restricted to this list.
In this way, the complexity of accessing a given coordinate $q_{i_j}$ is $O(\log k)$, rather than $O(\log d)$.

For simplicity in describing our data structures, since all our query algorithms choose coordinates from a set of at most $O(n^2)$ options (and thus $k=O(n^2)$), we assume we can access an entry of the query in time $O(\log n)$.

Since we often consider only a subset of coordinates of a vector, we use the notation $x_S$ to refer to $x$ restricted to the coordinates of $S$.
\begin{definition}[Subset Indexing]
    \label{def:subset-index}
    For any $x\in \mathbb{R}^{d}$ and $S\subseteq [d]$, where $S=\{i_1,\ldots ,i_k\}$ and $i_1 < i_2 <\ldots < i_k$, we write $x_S = (x_{i_1}, x_{i_2}, \ldots ,x_{i_k})$.
\end{definition}

We will use the following $\ell_\infty$ ANN data structure by \cite{indyk2001approximate}.

\begin{theorem}[\cite{indyk2001approximate}]
    \label{thm:i01}
    There exists a data-structure for $\ell_\infty$ ANN with
    \begin{itemize}
        \item Approximation $O(\log \log d)$;
        \item Query time $O(d\polylog n)$; and
        \item Storage space $O(n^{1.01}d)$
    \end{itemize}
\end{theorem}

\section{Subset Embeddings}
\label{sec:subset-embed}

In this section, we introduce the concept of subset embeddings, which play a crucial role in our ANN data structure constructions.
Standard dimensionality reduction techniques, such as the Johnson-Lindenstrauss transform, often require projecting data onto a dense set of new basis vectors.
However, in the ultra-high dimensional setting, reading even a single full-dimensional vector at query time is prohibitively expensive.
Thus, we seek embeddings that select a small subset of the original coordinates while preserving the distance structure of the dataset. This is similar to the notion of feature selection studied in the clustering and machine learning literature, but our goal is to preserve pairwise distances.

Recall from Definition \ref{def:subset-index} that we define $x_S$ to be the $|S|$-dimensional vector of $x$ at the coordinates given in $S$, for any $S\subseteq [d]$ and $x\in \mathbb{R}^{d}$.
\begin{definition}[$\ell_\infty$ Subset Embedding]
    Given $n$ points $X \subseteq \mathbb{R}^{d}$, a set $S\subseteq [d]$ induces an \textit{$\ell_\infty$ Subset Embedding} for $X$ with dimension $|S|$ and distortion $c\geq 1$ if
    \[
        \text{for all } x,y\in X,\ \|x_S - y_S\|_\infty \geq \frac{1}{c}\cdot \|x-y\|_\infty.
    \]
\end{definition}
Note that by definition of $\ell_\infty$, for any set $S$, $\|x_S - y_S\|_\infty \leq \|x - y\|_\infty$.
We often refer to $\ell_\infty$ subset embeddings simply as subset embeddings, as we only consider $\ell_\infty$ distance.

Like standard dimensionality reduction, subset embeddings enable efficient distance estimation between dataset points. However, in ANN applications, the query point is unknown during preprocessing when the subset embedding is constructed. The following lemma shows that distortion guarantees for subset embeddings extend to ANN approximation guarantees: an approximate nearest neighbor under the subset embedding remains a good approximation to the true nearest neighbor.

\begin{lemma}
    \label{lem:subset-to-ann}
    Consider $X\subset \mathbb{R}^{d}$, and suppose $S\subseteq [d]$ induces an $\ell_\infty$ Subset Embedding for $X$ with distortion $\alpha\geq 1$.
    Define $X_S = \{x_S \mid x\in X\}$, and fix some point $q\in \mathbb{R}^{d}$.
    Let $p^{*}\in X$ be the nearest neighbor (under $\ell_\infty$) to $q$, and consider any $\hat{p}\in X$ such that $\hat{p}_S$ is a $\beta$-approximate nearest neighbor to $q_S$ in $X_S$, for some $\beta \geq 1$.
    Then,
    \[
        \|q-\hat{p}\|_\infty \leq (\alpha(\beta + 1)+1)\|q-p^{*}\|_\infty
    \]
\end{lemma}
\begin{proof}
    Since $S$ is a subset embedding of $X$ with distortion $\alpha$,  $\|\hat{p} - p^{*}\|_\infty \leq \alpha \|\hat{p}_S - p^{*}_S\|_\infty$.
    In addition, since $\hat{p}_S$ is a $\beta$-approximate nearest neighbor to $q_S$ in $X_S$, for all $x\in X$, $\|\hat{p}_S-q_S\|_\infty \leq \beta \|x_{S}-q_S\|_\infty$.
    So, repeatedly applying the triangle inequality and with $r=\|p^{*} - q\|_\infty$,
    \begin{align*}
        \|\hat{p} - q\|_\infty &\leq \|\hat{p}-p^{*}\|_\infty + \|p^{*}-q\|_\infty \\
                             &\leq \alpha \|\hat{p}_S - p^{*}_S\|_\infty + r \\
                             &\leq \alpha \|\hat{p}_S - q_S \|_\infty + \alpha \|q_S - p^{*}_S\|_\infty + r \\
                             &\leq \alpha\beta \|p^{*}_S - q_S\|_\infty + \alpha \|q_S - p^{*}_S\|_\infty + r \\
                             &=\alpha(\beta + 1)\|p^{*}_S - q_S\|_\infty + r \\
                             &\leq (\alpha(\beta + 1) + 1)r
    \end{align*}
    with the final inequality following from $\|q_S - p^{*}_S\|_\infty \leq \|q-p^{*}\|_\infty = r$.
\qedhere

\end{proof}

\subsection{Results for \texorpdfstring{$\ell_\infty$}{} Subset Embeddings}
In this subsection, we give some simple constructions and lower bounds for $\ell_\infty$ subset embeddings.
We show that there is an isometric subset embedding with $n(n-1)/2 = O(n^2)$ dimensions, and that no subset embedding with distortion less than 2 can have smaller dimension.
We also show that $\Omega(n)$ dimensions are necessary for any finite distortion.
Finally, we state Theorem \ref{thm:subset-embed-imp}, our main technical result for subset embeddings, which shows that for any $c\geq 1$, there exists a subset embedding with distortion $O(c)$ and dimension $n^{1+1/c}$.
We give the full construction and proof of Theorem \ref{thm:subset-embed-imp} in the next subsection.
\begin{lemma}
    \label{lem:isometric-subset}
    For any $n$ point data set $X \subseteq \mathbb{R}^{d}$, there exists an $\ell_\infty$ subset embedding for $X$ with dimension at most $n(n-1)/2 = O(n^2)$ and distortion 1.
\end{lemma}
\begin{proof}
    Construct the set $S$ which induces the embedding in the following way.
    Initialize $S=\emptyset$.
    Select any two points $u,v\in X$ such that $\|u_S - v_S\|_\infty < \|u-v\|_\infty$, and add to $S$ a coordinate $i$ such that $|u_i - v_i| = \|u - v\|_\infty$; by definition of $\ell_\infty$, such a coordinate always exists.
    Repeat until for all $x,y\in X$, $\|x_S - y_S\|_\infty = \|x-y\|_\infty$.

    By construction, the embedding has distortion 1.
    Moreover, for each pair $x,y\in X$, at most one coordinate is added to $S$.
    Since there are $\binom{n}{2}=n(n-1)/2$ pairs of points, it follows that $|S|\leq n(n-1)/2$.
\end{proof}

If one desires distortion less than 2, it turns out the simple isometric embedding of Lemma \ref{lem:isometric-subset} is tight.
\begin{lemma}
    \label{lem:lb-dist-2}
    There exist $n$ points $X\subseteq \mathbb{R}^{d}$ such that any $\ell_\infty$ subset embedding of $X$ with distortion $c < 2$ must have dimension at least $n(n-1)/2$.
\end{lemma}
\begin{proof}
    Set $d=n(n-1)/2$, and assign for each $j>i$ a unique index $m(i,j)$ from $[d]$.
    The construction of $X=\{x^{(1)},\ldots ,x^{(n)}\}$ is as follows.
    For each $i\in [n]$, vector $x^{(i)} \in \{0,1,2\}^{d}$ is
    \begin{itemize}
        \item 0 at coordinate $m(i,j)$, for all $j > i$;
        \item 2 at coordinate $m(k,i)$, for all $k < i$;
        \item 1 elsewhere.
    \end{itemize}
    So, by construction, for any $i,j\in [n]$ with $i < j$, $\left\|x^{(i)} - x^{(j)}\right\|_\infty = \left|x^{(i)}_{m(i,j)} - x^{(j)}_{m(i,j)}\right| = 2$.

    Consider any set $S\subset [d]$ of size $|S| < d$.
    There must be some $i,j$ such that $m(i,j)\not\in S$, since there are $n(n-1)/2$ pairs $i<j$ and $d=n(n-1)/2$.
    But then, $\left\|x^{(i)}_S - x^{(j)}_S\right\|_\infty \leq 1$, since at all coordinates in $S$, no coordinate has $\left|x^{(i)}_k - x^{(j)}_k\right| = 2$.
\end{proof}

In addition, unlike in standard embeddings, to obtain any finite distortion, one must use at least $\Omega(n)$ coordinates. 
\begin{lemma}
    \label{lem:lb-finite-dist}
    There exists a dataset $X\subseteq \mathbb{R}^{d}$ of $n$ unique points and $d \geq n$ such that for any subset $S\subseteq [d]$ of size at most $n-2$, there are distinct $x,y\in X$ such that $x_S = y_S$.
\end{lemma}
\begin{proof}
    Let $X = \{e^{i}\}_{i\in [n]}$, where $e^{i}\in \mathbb{R}^{n}$ is the vector with a 1 at coordinate $i$ and 0 elsewhere.
    Let $S\subset [n]$ be any set with $|S|\leq n-2$; so, there exist distinct $i,j\in [n]$ such that $i,j\not\in S$.
    Since $i\not\in S$, it follows that $e^{i}_S = \vec{0}$; similarly, $j\not\in S$ and $e^{j}_S = \vec{0}$.
    So, $e^{i}_S = e^{j}_S$ but $e^{i}\neq e^{j}$.
\end{proof}

Nonetheless, by increasing the distortion, we obtain subset embeddings with near linear dimension.

\subsetemb*
In particular, Theorem \ref{thm:subset-embed-imp} implies that there exists an $\ell_\infty$ subset embedding for $X$ with distortion $O(\log n)$ and dimension $O(n)$, and distortion $O(1)$ with dimension $n^{1.01}$.

\subsection{Subset Embedding Construction}
\label{sec:subset-embed-construction}

We now prove Theorem \ref{thm:subset-embed-imp} constructively, by giving an algorithm which produces the desired embeddings.
The construction is greedy: as long as some pair of dataset points does not have its distance preserved to within the target distortion by the coordinates selected so far, we take the furthest such pair and add a coordinate realizing its distance.

Fix a distortion parameter $\alpha \geq 1$.
For a set $S\subseteq [d]$ of already-selected coordinates, we say that a pair $\{u,v\}\subseteq X$ is \textit{unresolved} with respect to $S$ if
\[
    \|u_S - v_S\|_\infty < \frac{1}{\alpha}\|u-v\|_\infty,
\]
using the convention that $\|u_\emptyset - v_\emptyset\|_\infty = 0$.
The routine $\se$, given in Figure \ref{fig:subset-embed-imp}, halts exactly when no unresolved pair remains, and so its output induces a subset embedding with distortion $\alpha$ by definition.
The first coordinate it selects realizes the diameter of $X$, and each later coordinate realizes the largest distance among the pairs which remain unresolved.

\begin{tcolorbox}[breakable, enhanced]
    \textbf{Constructing Subset Embeddings with the Routine \se}
    \medskip

    \textbf{Input:} A set $X\subset \mathbb{R}^{d}$ of $n$ points, distortion parameter $\alpha\geq 1$
    \medskip

    \textbf{Output:} A subset $S\subseteq [d]$ (which induces a subset embedding)
    \medskip

    \textbf{Procedure $\se(X,\alpha)$}:
    \begin{enumerate}
        \item Set $S_0 = \emptyset$.\vspace{4pt}
        \item For $\tau = 1,2,\ldots$:\vspace{4pt}
            \begin{enumerate}
                \item If no pair of $X$ is unresolved with respect to $S_{\tau-1}$, return $S = S_{\tau-1}$.\vspace{4pt}
                \item Let $e_\tau = \{u^{\tau},v^{\tau}\}$ be a pair which is unresolved with respect to $S_{\tau-1}$ and which maximizes $\|u^{\tau}-v^{\tau}\|_\infty$ among all such pairs, breaking ties arbitrarily.
                    Write $D_\tau = \|u^{\tau}-v^{\tau}\|_\infty$.\vspace{4pt}
                \item Let $i_\tau$ be a coordinate with $\left|u^{\tau}_{i_\tau} - v^{\tau}_{i_\tau}\right| = D_\tau$; such a coordinate always exists by definition of $\|\cdot\|_\infty$.\vspace{4pt}
                \item Set $S_\tau = S_{\tau-1}\cup \{i_\tau\}$.\vspace{4pt}
            \end{enumerate}
    \end{enumerate}
\end{tcolorbox}
\begin{figure}[H]
    \vspace{-10pt}
    \caption{The subset embedding construction for Theorem \ref{thm:subset-embed-imp}.}
    \label{fig:subset-embed-imp}
\end{figure}

Throughout the analysis, we write $m$ for the number of iterations performed by $\se(X,\alpha)$, so that it returns $S = S_m$, and we keep the notation $e_\tau, D_\tau, i_\tau, S_\tau$ of Figure \ref{fig:subset-embed-imp}.
Notice that $S_0 \subseteq S_1 \subseteq \cdots \subseteq S_m$, and that each $e_\tau$ is unresolved with respect to $S_{\tau-1}$ but resolved with respect to $S_\tau$, since $i_\tau\in S_\tau$ and thus $\left\|u^{\tau}_{S_\tau} - v^{\tau}_{S_\tau}\right\|_\infty = D_\tau$.

By the stopping rule of $\se$, its output always induces a subset embedding with distortion $\alpha$; it thus remains to bound the number of selected coordinates, which is at most $m$.
To do so, we view the selected pairs $e_1,\ldots ,e_m$ as the edges of a graph over the dataset, and show that the greedy order forces this graph to have no short cycles.
The key fact is Lemma \ref{lem:se-later-pairs}: any pair selected after iteration $s$ is both shorter than $e_s$ and close on the coordinate $i_s$.
So, if the selected pairs contained a short cycle, one could walk from one endpoint of the earliest-selected edge of the cycle to its other endpoint in a few steps, each of which moves only slightly on the coordinate $i_s$; but $i_s$ is precisely the coordinate on which these two endpoints are far apart.
It then suffices to bound the number of edges in a graph with no short cycles.

We first observe some simple properties of the routine.
\begin{lemma}
    \label{lem:se-basic}
    For any set $X\subset \mathbb{R}^{d}$ of $n$ points and any $\alpha\geq 1$, the pairs $e_1,\ldots ,e_m$ selected by $\se(X,\alpha)$ are distinct, and each consists of two distinct points of $X$.
    Consequently $\se(X,\alpha)$ halts after $m \leq \binom{n}{2}$ iterations and returns a set $S$ with $|S|\leq m$ which induces an $\ell_\infty$ subset embedding for $X$ with distortion $\alpha$.
\end{lemma}
\begin{proof}
    If $\{u,v\}$ is unresolved with respect to some set, then $\|u-v\|_\infty > \alpha\|u_S-v_S\|_\infty \geq 0$ and so $u\neq v$; in particular, each $e_\tau$ consists of two distinct points and $D_\tau > 0$.
    For distinctness, fix $\tau$ and any later iteration $t > \tau$.
    As noted above, $e_\tau$ is resolved with respect to $S_\tau$, and $S_\tau \subseteq S_{t-1}$, and so $e_\tau$ is resolved with respect to $S_{t-1}$ as well; but $e_{t}$ is unresolved with respect to $S_{t-1}$, and thus $e_t \neq e_\tau$.

    Since $X$ has only $\binom{n}{2}$ pairs of points, it follows that $m\leq \binom{n}{2}$ and the routine halts.
    It returns $S = S_m = \{i_1,\ldots ,i_m\}$, so $|S|\leq m$, and it does so only when no pair of $X$ is unresolved with respect to $S$, i.e.~when $\|u_S - v_S\|_\infty \geq \frac{1}{\alpha}\|u-v\|_\infty$ for all $u,v\in X$.
    This is exactly the statement that $S$ induces an $\ell_\infty$ subset embedding for $X$ with distortion $\alpha$.
\end{proof}

\begin{definition}[Witness Graph]
    \label{def:witness-graph}
    The \textit{witness graph} of an execution of $\se(X,\alpha)$ in \Cref{fig:subset-embed-imp} is the graph $G$ over vertex set $X$ with edge set $\{e_1,\ldots ,e_m\}$.
\end{definition}
By Lemma \ref{lem:se-basic}, the selected pairs are distinct pairs of distinct points, and so $G$ is a simple graph with exactly $m$ edges.
Moreover, each edge of $G$ is selected at exactly one iteration, and so we may speak of one edge of $G$ being selected before another.

The following lemma is where the greedy order is used: it shows that the coordinate $i_s$ selected at iteration $s$ is a coordinate on which every later selected pair has a small difference.
\begin{lemma}
    \label{lem:se-later-pairs}
    For any two iterations $s < t$ of $\se(X,\alpha)$,
    \[
        \left|u^{t}_{i_s} - v^{t}_{i_s}\right| < \frac{D_s}{\alpha}.
    \]
\end{lemma}
\begin{proof}
    The pair $e_t$ is unresolved with respect to $S_{t-1}$, and $S_{s-1}\subseteq S_{s}\subseteq S_{t-1}$ because $s\leq t-1$.
    Since restricting to a smaller set of coordinates can only decrease an $\ell_\infty$ distance,
    \begin{equation}
        \label{eq:se-later-unresolved}
        \left\|u^{t}_{S_{s-1}} - v^{t}_{S_{s-1}}\right\|_\infty \leq \left\|u^{t}_{S_{s}} - v^{t}_{S_{s}}\right\|_\infty \leq \left\|u^{t}_{S_{t-1}} - v^{t}_{S_{t-1}}\right\|_\infty < \frac{D_t}{\alpha};
    \end{equation}
    that is, $e_t$ is unresolved with respect to both $S_{s-1}$ and $S_s$.

    Since $e_t$ is unresolved with respect to $S_{s-1}$, it was one of the candidates available at iteration $s$, and iteration $s$ selects a candidate of maximum $\ell_\infty$ distance; thus, $D_t \leq D_s$.
    As $i_s \in S_s$, the middle term of \eqref{eq:se-later-unresolved} bounds the difference at coordinate $i_s$, and so
    \[
        \left|u^{t}_{i_s} - v^{t}_{i_s}\right| \leq \left\|u^{t}_{S_s} - v^{t}_{S_s}\right\|_\infty < \frac{D_t}{\alpha} \leq \frac{D_s}{\alpha}. \qedhere
    \]
\end{proof}

We now show that the witness graph has large girth.
\begin{lemma}
    \label{lem:witness-girth}
    The witness graph of any execution of $\se(X,\alpha)$ has girth greater than $\alpha+1$.
\end{lemma}
\begin{proof}
    Suppose, for contradiction, that the witness graph $G$ contains a cycle $w_1w_2\cdots w_g w_1$ of length $g \leq \alpha+1$.
    Since $G$ is simple, $g\geq 3$, and the $g$ edges $\{w_j,w_{j+1}\}$ of the cycle are distinct (with indices modulo $g$, so that $w_{g+1}=w_1$); thus, they are selected at $g$ distinct iterations.
    Without loss of generality, by rotating the labels, suppose $\{w_1,w_2\}$ is the earliest selected of these edges, say $\{w_1,w_2\} = e_s$, and write $i = i_s$ and $D = D_s = \|w_1 - w_2\|_\infty$.

    Each of the other $g-1$ edges of the cycle is selected at some iteration $t > s$, and so by Lemma \ref{lem:se-later-pairs}, $|w_{j,i} - w_{j+1,i}| < D/\alpha$ for all $j\in \{2,\ldots ,g\}$.
    On the other hand, the coordinate $i$ was selected to realize the distance of $e_s$, and so $|w_{1,i}-w_{2,i}| = D$.
    So, summing these $g-1$ differences along the rest of the cycle from $w_2$ back to $w_1$ gives the contradiction:
    \[
        D = \left|w_{1,i} - w_{2,i}\right| \leq \sum_{j=2}^{g}\left|w_{j,i}-w_{j+1,i}\right| < (g-1)\cdot \frac{D}{\alpha} \leq \frac{\alpha}{\alpha}D = D,
    \]
    with the first inequality following triangular inequality, and the final inequality following from $g\leq \alpha+1$.
\end{proof}

It remains to convert the girth bound into a bound on the number of edges of the witness graph, which is exactly the number of iterations $m$.
\begin{lemma}[{\protect\cite{alon2002moore}}]
    \label{lem:moore-bound}
    For any integer $k\geq 1$, every graph $H$ on $n$ vertices with girth at least $2k+1$ has at most $\frac{1}{2}\left(n^{1+1/k}+n\right)$ edges.
\end{lemma}
\begin{proof}
    Let $m$ be the number of edges of $H$ and $\bar{d}=2m/n$ its average degree.
    If $\bar{d} < 2$, then $m < n \leq \frac{1}{2}(n^{1+1/k}+n)$ and we are done, so suppose $\bar{d}\geq 2$; notice that this forces $m\geq n$, and thus $H$ contains a cycle and has finite girth.
    \cite{alon2002moore} show that a graph on $n$ vertices with average degree $\bar{d}\geq 2$ and girth $2r+1$ satisfies $n\geq 1 + \bar{d}\sum_{t=0}^{r-1}(\bar{d}-1)^{t}$, and that a graph with girth $2r$ satisfies $n \geq 2\sum_{t=0}^{r-1}(\bar{d}-1)^{t}$.
    If $H$ has girth at least $2k+1$, then $r\geq k$ in the odd case and $r\geq k+1$ in the even case.
    So, keeping only a single term of each sum and using $\bar{d} > \bar{d}-1\geq 1$, these bounds read, respectively,
    \[
        n \geq 1 + \bar{d}\left(\bar{d}-1\right)^{k-1} > \left(\bar{d}-1\right)^{k}
        \qquad\text{and}\qquad
        n \geq 2\left(\bar{d}-1\right)^{k} > \left(\bar{d}-1\right)^{k}.
    \]
    In either case, $n > (\bar{d}-1)^{k}$, and so $\bar{d} < n^{1/k}+1$, giving $m = \frac{n\bar{d}}{2} < \frac{1}{2}\left(n^{1+1/k}+n\right)$.
\end{proof}

\begin{proof}[Proof of Theorem \ref{thm:subset-embed-imp}]
    Let $k = \lceil c\rceil$, run $\se(X,\alpha)$ with $\alpha = 2k-1$, and let $S$ be its output.
    By Lemma \ref{lem:se-basic}, $S$ induces an $\ell_\infty$ subset embedding for $X$ with distortion $\alpha = 2\lceil c\rceil - 1 \leq 2c+1$, and $|S|\leq m$, where $m$ is the number of iterations.

    The witness graph $G$ of this execution is a simple graph over the $n$ points of $X$ with exactly $m$ edges.
    By Lemma \ref{lem:witness-girth}, it has girth greater than $\alpha+1 = 2k$; since the girth is an integer, it is thus at least $2k+1$.
    So, by Lemma \ref{lem:moore-bound} and using $k = \lceil c\rceil \geq c$,
    \[
        |S| \leq m \leq \frac{1}{2}\left(n^{1+1/k}+n\right) \leq \frac{1}{2}\left(n^{1+1/c}+n^{1+1/c}\right) = n^{1+1/c}. \qedhere
    \]
\end{proof}

We conclude by arguing that the routine $\se$ is efficient.
Computing $\|x-y\|_\infty$ and a coordinate realizing it for all $\binom{n}{2}$ pairs takes $O(n^2 d)$ time, after which the pairs can be sorted by distance and processed in decreasing order, maintaining $\|u_S - v_S\|_\infty$ for each surviving pair in $O(1)$ time per selected coordinate.
The total running time is thus $O(n^2 d + n^2|S|) = O(n^2 d + n^{3+1/c})$, which is $O(n^2 d)$ in the ultra-high dimensional regime $d = \Omega(n^{1+1/c})$.

\subsection{Lower Bound}
\label{sec:subset-embed-lower}
We prove a lower bound on the size of subset embeddings with distortion at most $c$.
Like the upper bound construction, the lower bound relies on a bound on the number of edges in a graph of girth $\lambda$.
\begin{lemma}[{\protect\cite{erdos1963regulare}}]
    \label{lem:girth-lower}
    For any integer $\lambda\geq 3$ and $n$ sufficiently large, there exists a graph over $n$ nodes with girth at least $\lambda$ and $\Omega\left(n^{1+1/(\lambda-2)}\right)$ edges.
\end{lemma}

The construction is then as follows.
Let $G=(V,E)$ be a graph over $n$ nodes with girth $\lambda\geq c+2$.
For each $e\in E$, fix an arbitrary endpoint $u_e$ of $e$ and define the potential function over nodes $f_e(z) = d_{G-e}(z,u_e)$,\footnote{With the convention that if disconnected, $d_{G-e}(z,u_e) = n$.} where $G-e$ is the graph with edge $e$ removed.
Then, with $E=\{e_1,\ldots ,e_m\}$ arbitrarily ordered, for each node $v\in V$, define the vector (in $\mathbb{R}^{m}$) $x^{v} = (f_{e_1}(v), f_{e_2}(v), \ldots , f_{e_m}(v))$.
Let $X = \{x^{v} \mid v\in V\}$.

We first make some simple observations about $f$.
\begin{lemma}
    \label{lem:potential-facts}
    Fix any edge $e=(u,v)$ in $E$.
    Then:
    \begin{enumerate}
        \item $|f_e(u) - f_e(v)| \geq \lambda - 1$
        \item For all edges $h\neq e$, $|f_h(u) - f_h(v)| \leq 1$.
    \end{enumerate}
\end{lemma}
\begin{proof}
    Without loss generality, suppose $u_e = u$.
    Then, $f_e(u) = 0$ and $|f_e(u) - f_e(v)| = f_e(v)$.
    Since $G$ has girth $\lambda$, every cycle has length at least $\lambda$; thus, in $G-e$, $u$ and $v$ must have distance at least $\lambda-1$, as otherwise $e$ would lie in a cycle of length at most $\lambda-1$ in $G$.
    So, $f_e(v) \geq \lambda-1$, showing the first claim.

    For the second claim, fix an edge $h\neq e$.
    Since $h \neq e$, the edge $e$ is present in $G-h$, and so $u$ and $v$ are adjacent in $G-h$.
    Thus, by the triangle inequality\footnote{Again using the convention that disconnected nodes have distance $n$},
    \[
        |d_{G-h}(u) - d_{G-h}(v)| \leq d_{G-h}(u,v) \leq 1
    \]
    as desired.
\end{proof}

With this, we now prove that to preserve pairwise distances in $X$ up to a factor of $c$, a subset embedding must include all $m$ coordinates.
\begin{lemma}
    \label{lem:lower-bound-construction}
    Let $S\subset [m]$ have at most $m-1$ elements.
    Then, there exists some $x^{u}, x^{v}\in X$ for which $\|x^{v}_S - x^{u}_S\|_\infty < (1/c)\|x^{v} - x^{u}\|_\infty$.
\end{lemma}
\begin{proof}
    Since $|S| \leq m-1$, there is some index $k\in [m]\setminus S$; let $e = e_k = (u,v)$ be the corresponding edge of $G$, and consider the two points $x^{u},x^{v}\in X$.

    By the first part of Lemma \ref{lem:potential-facts}, the $k$-th coordinate alone gives
    \[
        \|x^{v} - x^{u}\|_\infty \geq |f_{e}(u) - f_{e}(v)| \geq \lambda - 1 \geq c+1,
    \]
    using that the graph $G$ was chosen with girth $\lambda \geq c+2$.
    On the other hand, every $h\in S$ indexes an edge $e_h \neq e$, as $k\not\in S$, and so by the second claim of Lemma \ref{lem:potential-facts},
    \[
        \|x^{v}_S - x^{u}_S\|_\infty = \max_{h\in S}\left| f_{e_h}(u) - f_{e_h}(v) \right| \leq 1.
    \]
    Thus,
    \[
        \|x^{v}_S - x^{u}_S\|_\infty \leq 1 < \frac{c+1}{c} \leq \frac{1}{c}\|x^{v}-x^{u}\|_\infty
    \]
    as desired.
\end{proof}

Theorem \ref{thm:subset-embed-lower-bound} then follows from Lemma \ref{lem:lower-bound-construction} and from Lemma \ref{lem:girth-lower}, applied with $\lambda = \lceil c \rceil + 2$.

\section{ANN Data-Structure Constructions}\label{sec:algorithms}
In this section, we present our main data structure constructions for approximate nearest neighbor search in $\ell_\infty$.
Our approach leverages the subset embeddings developed in the previous section to reduce the dimensionality of the problem while preserving distance.
We begin by presenting a simple 3-approximation algorithm that serves as a building block for our more advanced constructions.

\begin{lemma}
    \label{lem:3-approx}
    Given a dataset $X\subset \mathbb{R}^{d}$ of $n$ points, there exists a data structure with storage space $O(n^2\log d)$ which, upon receiving a query which is a tuple $(q,\calS)$ with $q\in \mathbb{R}^{d}$, $\calS \subseteq X$, outputs a 3-approximate nearest neighbor (over $\ell_\infty$) to $q$ in $\calS$.
    Moreover, the query time is $O(|\calS|^2)$.
\end{lemma}
\begin{proof}
    For all pairs $x,y\in X$, define $i_{xy}$ to be a coordinate such that $|x_{i_{xy}} - y_{i_{xy}}|=\|x-y\|_\infty$ (which always exists by definition of $\|\cdot\|_\infty$).
    To construct the data structure, store $i_{xy}$ along with the values ${x_{i_{xy}}}, y_{i_{xy}}$ for all $x,y\in X$.
    This requires space $O(n^2\log d)$.

    \newcommand{\csp}{\calI^{\calS}}
    For a subset $\calS \subseteq X$ and for each $x\in \calS$, define $\csp_x = \{i_{xy} \mid y\in \calS\}$ (that is, the coordinates needed to determine distances between $x$ and all other points in $\calS$).
    At query time, given a point $q\in \mathbb{R}^{d}$ and subset $\calS\subseteq X$, for each $x\in \calS$ compute
    \[
        \ddec(x,q) = \max_{i\in \csp_x}|x_i - q_i|.
    \]
    Output the $x\in \calS$ which minimizes $\ddec(x,q)$.

    Notice that, by construction, $|\csp_x| \leq |\calS|$ for all $x\in \calS$.
    So, for each $x\in \calS$, computing $\ddec(x,q)$ takes $O(|\calS|)$ time, and thus finding the point $x\in \calS$ which minimizes $\ddec(x,q)$ takes $O(|\calS|^2)$ time.

    It remains to show that the output point is a 3-approximate nearest neighbor.
    Given $q\in \mathbb{R}^{d},\calS \subseteq X$, let $p^{*}$ be the nearest neighbor to $q$ in $\calS$, and let $\hat{p}$ be the point which minimizes $\ddec(x, q)$.
    Let $i^{*}=i_{p^{*}\hat{p}}$ be the stored coordinate such that $|p^{*}_{i^{*}} - \hat{p}_{i^{*}}| = \|p^{*} - \hat{p}\|_\infty$; so, $i^{*}\in \csp_{\hat{p}}$.
    We then have
    \begin{align}
        \label{eq:ps-ph-i}\|p^{*} - q\|_\infty \geq \ddec(p^{*}, q) \geq \ddec(\hat{p}, q) \geq |\hat{p}_{i^{*}} - q_{i^{*}}|
    \end{align}
    as $\hat{p}$ is the point which minimizes $\ddec(x,q)$ and for all $x\in \calS$, $\|x-q\|_\infty\geq \ddec(x,q)$.

    So, applying the triangle inequality,
    \begin{align*}
        \|\hat{p} - q\|_\infty &\leq \|\hat{p} - p^{*}\|_\infty + \|p^{*} - q\|_\infty \\
                               &= |\hat{p}_{i^{*}} - p^{*}_{i^{*}}| + \|p^{*} - q\|_\infty \\
                               &\leq |\hat{p}_{i^{*}} - q_{i^{*}}| + |q_{i^{*}} - p^{*}_{i^{*}}| + \|p^{*} - q\|_\infty \\
                               &\leq 3\|p^{*} - q\|_\infty
    \end{align*}
    with the final inequality following from \eqref{eq:ps-ph-i} and $|p^{*}_{i^{*}} - q_{i^{*}}|\leq \|p^{*} - q\|_\infty$ by definition of $\|\cdot\|_\infty$.
\end{proof}

Lemma \ref{lem:3-approx} then immediately implies a data structure for 3-approximate nearest neighbor under $\ell_\infty$, matching the approximation lower bound of Theorem \ref{thm:lb-randomized-alg}.
\threeapprox*

    \subsection{ANN from Subset Embeddings and Proving Theorem \ref{thm:small-query-time}}
\label{sec:subset-ann}
Given a subset embedding $S$ for $X$, we can easily construct a simple data structure in the following way.
Create a new dataset $X_S = \{x_S \mid x\in X\}$ consisting of the points in $X$ restricted to the coordinates of $S$, and build a data structure for $X_S$ using a standard $\ell_\infty$-ANN data structure (from Theorem \ref{thm:i01}).
For a query $q$, construct $q_S$ (which only requires reading $q$ at the coordinates of $S$), and use the constructed data structure to find the approximate nearest neighbor to $q_S$ in $X_S$, which is the output.

Using the subset embeddings of Theorem \ref{thm:subset-embed-imp} and the data structure of \cite{indyk2001approximate} (Theorem \ref{thm:i01}), this naive approach yields an $O(c\log \log |S|) = O(c\log\log n)$ approximation in $O(|S|\polylog n) = \ot{n^{1+1/c}}$ query time.\footnote{
In addition to the bounds from Theorem \ref{thm:subset-embed-imp} and Theorem \ref{thm:i01}, the approximation guarantee follows from Lemma \ref{lem:subset-to-ann}}
Unfortunately, it also requires space\footnote{The space is actually even larger: $\Theta(n^{2.01+1/c}\log d)$.} $\Omega(|S|\cdot n\cdot \log d)=\Omega(n^{2+1/c}\log d)$, which exceeds the desired $O(n^2\log d)$ space in Theorem \ref{thm:small-query-time}.

To reduce the space, we partition the dataset into multiple \textit{groups}, determine an (approximate) nearest-neighbor from each group using the data structure of \Cref{thm:i01}, and then find an approximate nearest-neighbor from these ``candidates.''
The full data structure is given in Figure \ref{fig:subset-embed-ANN}.

\begin{tcolorbox}[breakable, enhanced]
    \textbf{Data Structure of Theorem \ref{thm:small-query-time}} \\
    \textbf{Input:} Dataset $X\subset \mathbb{R}^{d}$ of $n$ points, parameter $c\geq 1$.

    \medskip
    \textbf{Data structure construction.}
    \begin{enumerate}
        \item Set $\rho=(1-1/c)/2$.
        \item Partition $X$ into $n^{1-\rho}$ groups each containing $n^{\rho}$ points (arbitrarily).
        \item For each group $G$, construct a subset embedding $S_G$ for $G$ with distortion $O(c)$ and dimension $n^{(1+1/c)\rho}$, using Theorem \ref{thm:subset-embed-imp}.
        \item For each group $G$, compute $G_S=\{x_{S_G} \mid x\in G\}$ and construct a data structure with approximation $O(\log \log |S_G|)$ for $G_S$ using \Cref{thm:i01}.
        \item Store the data structure of Lemma \ref{lem:3-approx} for $X$.
    \end{enumerate}

    \medskip
    \textbf{Query Processing.}  Upon receiving a query $q\in \mathbb{R}^{d}$:
    \begin{enumerate}
        \item For each group $G$, compute $q_{S_G}$ and find an approximate nearest-neighbor $p^{G}$ in $G_S$, using the stored data structure of \Cref{thm:i01}.
        \item Let $\calC = \{p^{G}\}_G$ be the set of candidate points.
        \item Use the data structure of Lemma \ref{lem:3-approx} to compute a $3$-approximate nearest-neighbor $\hat{p}$ to $q$ from $\calC$.
        \item Output $\hat{p}$.
    \end{enumerate}
\end{tcolorbox}
\begin{figure}[H]
    \vspace{-10pt}
    \caption{The data structure for Theorem \ref{thm:small-query-time}}
    \label{fig:subset-embed-ANN}
\end{figure}

\begin{lemma}
    \label{lem:subset-ann-qt}
    The query time of the data structure of Figure \ref{fig:subset-embed-ANN} is $\ot{n^{1+1/c}}$.
\end{lemma}
\begin{proof}
    Let $\rho = (1-1/c)/2$.
    For each of the $n^{1-\rho}$ groups $G$, the algorithm computes $q_{S_G}$, where $|S_G| = n^{(1+1/c)\rho}$.
    Computing all such $q_{S_G}$ takes time $O(n^{1-\rho + (1+1/c)\rho}) = O(n^{1+\rho/c}) =O(n^{1+1/c})$.

    In order to construct the set of candidates $\calC$, for each of the $n^{1-\rho}$ groups $G$, the algorithm queries the data structure of \Cref{thm:i01} constructed for $G$.
This takes time $O(n^{1-\rho} \cdot n^{(1+1/c)\rho}\cdot \polylog n) = \ot{n^{1+1/c}}$, from \Cref{thm:i01} and the fact that for each group $G$, the subset embedding $S_G$ has $|S_G|=n^{(1+1/c)\rho}$.

Finally, finding a 3-approximate nearest neighbor from the set of candidates $\calC$ requires time $O(n^{2(1-\rho)}) = O(n^{1+1/c})$ using Lemma \ref{lem:3-approx}.
    So, the total time is $\ot{n^{1+1/c}}$.
\end{proof}

\begin{lemma}
    \label{lem:subset-ann-space}
    The space of the data structure of Figure \ref{fig:subset-embed-ANN} is $O(n^2 \log d)$.
\end{lemma}
\begin{proof}
    Let $\rho=(1-1/c)/2$.
    For each of the $n^{1-\rho}$ groups, the data structure stores a subset embedding of size $n^{(1+1/c)\rho}$; this requires space $n^{1-\rho}\cdot n^{(1+1/c)\rho}\cdot \log d = O(n^{1+1/c}\log d)$.

    For each of the $n^{1-\rho}$ groups, we store a data structure of Theorem \ref{thm:i01} for the $n^{\rho}$ points of the group.
    Since the subset embedding for each group has dimension $n^{(1+1/c)\rho}$, each of these data structures requires space $O(n^{(1+1/c)\rho} \cdot n^{1.01\rho})=O(n^{(1+1/c)\rho} \cdot n^{2\rho})$.
    So, across the $n^{1-\rho}$ groups, the total space is
    \[
        O(n^{1 - \rho + (1+1/c)\rho + 2\rho}) = O(n^{1+2\rho + \rho/c}) = O(n^{2 - 1/c + \rho/c}) = O(n^2)
    \]
    where the second equality uses $\rho =(1-1/c)/2$ and thus $1 + 2\rho + \rho/c = 1 + (1-1/c) + (1-1/c)/(2c) = 2 - 1/c + \rho/c$.
    Finally, the space needed to store the data structure of Lemma \ref{lem:3-approx} is $O(n^2\log d)$, and thus the total space is $O(n^2\log d)$.
\end{proof}

\begin{lemma}
    \label{lem:subset-ann-qual}
    When constructed on $X\subset \mathbb{R}^{d}$ of size $n$ and run on any query $q\in \mathbb{R}^{d}$, the data structure of Figure \ref{fig:subset-embed-ANN} returns an $O(c\log\log n)$-approximate nearest neighbor to $q$ in $X$.
\end{lemma}
\begin{proof}
    Let $p^{*}\in X$ be the true nearest neighbor to $q$, and let $G$ be the group which contains $p^{*}$.
    Let $S_G$, $G_S$ be constructed as in Figure \ref{fig:subset-embed-ANN} for $G$, and let $v$ be returned by the call to the data structure of \Cref{thm:i01} for $G_S$.
    As the subset embedding $S_G$ has distortion $O(c)$ and the data structure of \Cref{thm:i01} has approximation $O(\log \log |S_G|) = O(\log \log n)$, by Lemma \ref{lem:subset-to-ann}, $\|v-q\|_\infty = O(c\log\log n) \cdot \|p^{*}-q\|_\infty$.

    By construction of the set of candidates $\calC$, $v\in \calC$.
    Let $\hat{p}$ be a 3-approximate nearest neighbor to $q$ in $\calC$ (which Figure \ref{fig:subset-embed-ANN} obtains via the data structure of Lemma \ref{lem:3-approx}).
    So, for all $x\in \calC$, $\|\hat{p}-q\|_\infty \leq 3\|x-q\|_\infty$, and thus
    \[
        \|\hat{p} - q\|_\infty \leq 3\|v - q\|_\infty = O(c\log\log n)\cdot \|p^{*}-q\|_\infty
    \]
    as desired.
\end{proof}

    \subsection{Proof of \texorpdfstring{\Cref{thm:c-258}}{Theorem \ref{thm:c-258}}} \label{sec:c-258}

We now present the data structure for Theorem \ref{thm:c-258}, which relies on repeated applications of the 3-approximate nearest neighbor data structure of Lemma \ref{lem:3-approx}.
At each step, we partition the current set of candidate points into groups of size $n^{\alpha}$.
We then run the algorithm of Lemma \ref{lem:3-approx} on each group to find a 3-approximate nearest neighbor to the query from that group, and repeat until only a small number of candidates remain (on which we can directly apply Lemma \ref{lem:3-approx} again).
By increasing $\alpha$ proportionally as the number of points (and thus the number of groups) decreases, we ensure that the query time at each step is $O(n^{1+1/c})$.
At each step, the value of $\alpha$ doubles, starting from $\alpha = 1/c$, so there are $O(\log c)$ many steps in total.
Since the approximation increases by a factor of 3 at each step (from the use of Lemma \ref{lem:3-approx}), the final approximation is $O(3^{\log c}) = O(c^{\log 3})$, as desired.

The full data structure is given in Figure \ref{fig:c-258}, and we prove the approximation and query time in Lemma \ref{lem:c-258-approx}.

\begin{tcolorbox}[breakable, enhanced]
    \textbf{Data Structure of \Cref{thm:c-258}} \\
    \textbf{Input:} Dataset $X \subset \mathbb{R}^{d}$ of $n$ points, parameter $c\geq 1$.

    \medskip
    \textbf{Data structure construction.}
    \begin{enumerate}
        \item For all $x,y\in X$, store a coordinate $i$ such that $|x_i - y_i| = \|x-y\|_\infty$ along with the values $x_i,y_i$.
    \end{enumerate}

    \medskip
    \textbf{Query Processing.} Upon receiving a query $q\in \mathbb{R}^{d}$:
    \begin{enumerate}
        \item Initialize $\calA = X$.
        \item For $i = 1, 2, \dots, \lfloor \log c \rfloor$:
            \begin{enumerate}
                \item Set $\alpha_i = 2^{i-1}/c$, and arbitrarily partition $\calA$ into $L = \lceil|\calA| / n^{\alpha_i}\rceil$ many groups $G_1,\ldots ,G_L$, each of size at most $n^{\alpha_i}$.
                \item For each group $G_j$, run the algorithm of Lemma \ref{lem:3-approx} on $G_j$ and query $q$ to find a $3$-approximate nearest neighbor $\hat{p}_j$ in $G_j$ to $q$.
                \item Set $\calA = \{\hat{p}_1, \hat{p}_2, \dots, \hat{p}_L\}$.
            \end{enumerate}
        \item Run the algorithm of Lemma \ref{lem:3-approx} on $\calA$ and query $q$ to obtain a final point $\hat{p}$. Return $\hat{p}$.
    \end{enumerate}
\end{tcolorbox}
\begin{figure}[H]
    \vspace{-10pt}
    \caption{The data structure for \Cref{thm:c-258}}
    \label{fig:c-258}
\end{figure}

It is immediate that the data structure of Figure \ref{fig:c-258} uses space $O(n^2 \log d)$.
It thus remains to show the query time and approximation factor.

\begin{lemma}
    \label{lem:c-258-query-time}
    The query time of the data structure of Figure \ref{fig:c-258} is $O(n^{1+1/c}\log c)$.
\end{lemma}
\begin{proof}
    At each step $i$, the query algorithm partitions the current set $\calA$ into groups of size $n^{\alpha_i} = n^{2^{i-1}/c}$.
    For each group, it runs the algorithm of Lemma \ref{lem:3-approx} to find a $3$-approximate nearest neighbor to $q$; this takes time $O(n^{2\alpha_i})$ per group.
    Since there are $O(|\calA| / n^{\alpha_i})$ many groups, the total time for this step is $O(|\calA| / n^{\alpha_i} \cdot n^{2\alpha_i}) = O(|\calA| \cdot n^{\alpha_i})$.

    Initially, $|\calA| = n$.
    At each step $i$, the number of points in $\calA$ reduces by a factor of $n^{\alpha_i}$, since we keep only one point from each group of size $n^{\alpha_i}$.
    So, at the beginning of step $i$, we have
    \[|\calA| = \frac{n}{\prod_{j=1}^{i-1} n^{\alpha_j}} = n^{1 - \sum_{j=1}^{i-1} \alpha_j} = n^{1 - \left(2^{i-1} - 1\right)/c}.\]
    Thus, at step $i$, the time taken is
    \[O\left(|\calA| \cdot n^{\alpha_i}\right) = O\left(n^{1 - (2^{i-1} - 1)/c} \cdot n^{2^{i-1}/c}\right) = O(n^{1 + 1/c}).\]
    Since there are $O(\log c)$ many steps, the total query time is $O(n^{1+1/c} \cdot \log c) = \ot{n^{1+1/c}}$, as desired.
\end{proof}

\begin{lemma}
    \label{lem:c-258-approx}
    When run on any query $q\in \mathbb{R}^{d}$, the data structure of Figure \ref{fig:c-258} returns a $3 \cdot c^{\log 3}$-approximate nearest neighbor to $q$ in $X$.
\end{lemma}
\begin{proof}
    Let $p^{*}\in X$ be the true nearest neighbor to $q$.
    We show by induction that after step $i$, the set $\calA$ contains a $3^i$-approximate nearest neighbor to $q$ in $X$.
    The base case $i=0$ is trivial, as initially $\calA = X$.
    Now, suppose after step $i-1$, $\calA$ contains a $3^{i-1}$-approximate nearest neighbor to $q$ in $X$; denote this point as $p^{(i-1)}$.
    Consider step $i$.
    Let $G$ be the group in step $i$ which contains $p^{(i-1)}$.
    Since the algorithm of Lemma \ref{lem:3-approx} returns a $3$-approximate nearest neighbor to $q$ in $G$,
    the set $\calA$ after step $i$ contains a point $p^{(i)}$ such that
    \[\|p^{(i)} - q\|_\infty \leq 3 \cdot \|p^{(i-1)} - q\|_\infty \leq 3^i \cdot \|p^{*} - q\|_\infty.\]
    Thus, after step $\lfloor \log c \rfloor$, the set $\calA$ contains a $3^{\lfloor \log c \rfloor}$-approximate nearest neighbor to $q$ in $X$.
    Finally, the last call to the algorithm of Lemma \ref{lem:3-approx} returns a $3$-approximate nearest neighbor to $q$ in $\calA$, which is a $3^{\lfloor \log c \rfloor + 1} \leq 3 \cdot 3^{\log c} = 3 \cdot c^{\log 3}$-approximate nearest neighbor to $q$ in $X$, as desired.
\end{proof}

Theorem \ref{thm:c-258} then follows from Lemmas \ref{lem:c-258-query-time} and \ref{lem:c-258-approx}.

\section*{Acknowledgement}
We thank Alexandr Andoni for suggesting the proof idea of Theorem \ref{thm:subset-embed-imp}.
We also thank Erik Waingarten for suggestions on an earlier manuscript which improved the presentation of the results.

\bibliographystyle{alpha}
\bibliography{ref.bib}

\appendix
\section{Proofs of Lower Bounds}
\label{sec:lower-bound-proofs}
In this section, we present lower bounds on the query time of $\ell_\infty$ ANN data structures in the ultra-high dimensional regime.
These results justify the necessity of higher approximation guarantees when aiming for query times independent of $d$ in the ultra-high dimensional regime.
We first prove the result for deterministic data structures and then extend it to randomized algorithms using Yao's Minimax Principle (which will then prove \Cref{thm:lb-randomized-alg}).

\begin{restatable}{lemma}{threelower}
    \label{lem:3-lower}
    For all $c < 3$, any deterministic data structure for $c$-approximate nearest neighbor over $\ell_\infty^d$ must have query time $\Omega(d)$.
\end{restatable}

\begin{proof}
    We prove the contrapositive: For any deterministic data structure for $c$-approximate nearest neighbor over $\ell_\infty^{d}$ with query time $o(d)$, there exist a dataset and a query such that the returned point is not a $c$-approximation.
Let the adversarial dataset $P \subset \mathbb{R}^d$ contain two points $a,b \in \mathbb{R}^d$ such that $a = x + \vec{1}, b = x - \vec{1}$ for an arbitrary vector $x \in \mathbb{R}^d$.
The query $q$ has entries $q_i = x_i$ except for some $i_0 \in [d]$ that is not queried by the deterministic query algorithm.
The query algorithm deterministically outputs either $a$ or $b$.
If it outputs $a$, let $q_{i_0} = x_{i_0} - 2$; then $a$ is not a $c$-approximation since $c < 3$.
Otherwise, let $q_{i_0} = x_{i_0} +2$; then $b$ is not a $c$-approximation.
\end{proof}

By simply applying Yao's Minimax Principle, we can further prove a lower bound result for randomized algorithms, which is exactly \Cref{thm:lb-randomized-alg}.

\begin{proof}[Proof of \Cref{thm:lb-randomized-alg}]
    Let $R_{\br}$ be a randomized algorithm over some distribution $\lambda$ of random strings $\br$ for $c$-approximate nearest neighbor over $\ell_{\infty}$ with success probability at least $0.5001$ for any dataset and query $q$.
We know, for any distribution $\mu$ over the queries,
    \begin{align*}
        0.5001 \leq & \sum_{\bq \in \supp(\mu)} p_\mu(\bq) \cdot \sum_{\br \in \supp(\lambda)} p_\lambda(\br) \ind{\text{$R_{\br}$ outputs a valid answer on $\bq$}} \\
        = & \sum_{\br \in \supp(\lambda)} p_\lambda(\br) \cdot \sum_{\bq \in \supp(\mu)} p_\mu(\bq) \ind{\text{$R_{\br}$ outputs a valid answer on $\bq$}} \\
        \leq & \max_r \sum_{\bq \in \supp(\mu)} p_\mu(\bq) \ind{\text{$R_{r}$ outputs a valid answer on $\bq$}}.
    \end{align*}
    Let $r^*$ be the optimizer of the last expression above.
Note $R_{r^*}$ is a deterministic data structure with success probability at least $0.5001$ over $\bq \sim \mu$, whose query time lower bounds that of $R_{\br}$.
Now we show there exists a dataset and a distribution of queries such that any deterministic data structure that answers at least a $0.5001$ fraction of queries correctly must have query time $\Omega(d)$, which will complete the proof.

    Let the dataset be $P$ in the above proof of \Cref{lem:3-lower} (i.e., $P$ contains two points $a = x + \vec{1}, b = x - \vec{1}$ where $x \in \mathbb{R}^d$ is an arbitrary vector).
Let the distribution of queries be defined as follows: (1) sample an index $i_0 \in [d]$ uniformly at random; (2) the query $q$ is such that $q_i = x_i$ for all $i \neq i_0$; (3) $q_{i_0} = x_{i_0} - 2$ with probability $0.5$, and $q_{i_0} = x_{i_0} + 2$ with probability $0.5$.

    Let $Q$ be the set of coordinates that $R_{r^*}$ checks.
    In order for $R_{r^*}$ to be correct with probability at least $0.5001$ over the above distribution of queries, we have
    \[\frac{|Q|}{d} + \left(1 - \frac{|Q|}{d}\right) \cdot \frac{1}{2} \geq 0.5001, \]
    which gives $|Q| \geq 0.0001 \cdot d  = \Omega(d)$ as desired.
\end{proof}

It is also straightforward to see that in any metric space, any ANN data structure with finite approximation must have query time $\Omega(\min\{n,d\})$.
\begin{lemma}
    \label{lem:qt-lower-bound}
    Under any metric space $(\mathbb{R}^{d}, d)$, there exists a dataset $X\subset \mathbb{R}^{d}$ of $n$ points for which any data structure with finite approximation and a deterministic query algorithm must have query time $\Omega(\min \{n,d\})$.
\end{lemma}
\begin{proof}
    Let $k = \min \{n,d\}$, and consider $X$ which contains $e_1,\ldots ,e_k$, where $e_i$ is the standard unit basis vector that is 1 at coordinate $i$ and 0 elsewhere.
    Let $A$ be any query algorithm which examines only $k/2$ elements of a query; since the algorithm is deterministic, it can thus output only $k/2$ distinct points across all queries.
    So, there exist query points $e_i,e_j$, $i\neq j\in [k]$ for which $A$ outputs the same point as a nearest neighbor; let this point be $p$ and without loss of generality suppose $e_i\neq p$.
    As $e_i \neq p$ and $d$ is a metric, $d(e_i, p) > 0$.
    However, $e_i \in X$ by construction, so the nearest neighbor to $e_i$ in $X$ has distance $d(e_i, e_i) = 0$.
    Thus, the approximation is unbounded.
\end{proof}

The same lower bound for randomized algorithms also follows from Yao's Minimax Principle, in a similar manner to the proof of Theorem \ref{thm:lb-randomized-alg}.

\end{document}